\documentclass[reqno]{amsart}
\usepackage{IEEEtrantools}
\usepackage{blindtext}
\usepackage{tasks}
\usepackage{graphicx}
\usepackage{mathrsfs}
\usepackage{amsmath,amssymb,mathrsfs,amsfonts,dsfont} 
\usepackage{breqn}
\usepackage{mathtools}
\usepackage{caption}
\usepackage{subcaption}
\usepackage{blindtext}
\usepackage{tasks}
\usepackage{graphicx}
\usepackage{mathrsfs}
\usepackage{amsmath,amssymb,mathrsfs,amsfonts,dsfont} 
\usepackage{breqn}
\usepackage{mathtools}
\usepackage{caption}
\usepackage{subcaption}
\mathtoolsset{showonlyrefs=true}
\newtheorem{theorem}{Theorem}[section]

\newtheorem{definition}[theorem]{Definition}

\numberwithin{equation}{section}

\newcommand{\modes}{\mathsf{P}}

\newcommand{\Real}{\mathbb{R}}
\newcommand{\Natural}{\mathbb{N}}
\newcommand{\Expectation}[1]{\ensuremath{\mathds{E}[#1]}}

\newcommand{\infgen}[1]{\mathcal{L}#1}
\newcommand{\define}{\coloneqq}
\newcommand{\Falgebra}{\ensuremath{\mathcal{F}}}
\newcommand{\Pprob}{\ensuremath{\mathds{P}}}
\newcommand{\ExtInputMap}{\ensuremath{\mathcal{U}}}
\newcommand{\IntInputMap}{\ensuremath{\mathcal{W}}}

\newcommand{\horizontalEqSep}{\noindent\rule{18cm}{0.4pt}
}
\newcommand{\expo}{\mathsf{e}}
\newcommand{\diff}{\mathsf{d}}

\newcommand{\norm}[1]{\lVert#1\rVert}

\mathtoolsset{showonlyrefs=true}
\newcommand{\PC}{\mathcal{P}}

\newcommand{\ssft}{ \textrm{SSF-}\textrm{M}_{2}}
\begin{document}
	%\maketitle
	\begin{abstract}
		In this work, we derive conditions under which abstractions of networks of stochastic hybrid systems can be constructed compositionally. Proposed conditions leverage the interconnection topology, switching randomly between $\PC$ different interconnection topologies, and the joint dissipativity-type properties of subsystems and their abstractions. The random switching of the interconnection is modelled by a Markov chain. In the proposed framework, the abstraction, itself a stochastic hybrid system (possibly with a lower dimension), can be used as a substitute of the original system in the controller design process.  %In addition, we show that for a class of stochastic hybrid systems, the compositional abstraction can still be constructed if those conditions are violated for some of the $\PC$ interconnection topologies, provided that an additional condition on the Markov chain parameters is satisfied. 
		Finally, we provide an example illustrating the effectiveness of the proposed results by designing a controller enforcing some logic properties over the interconnected abstraction and then refining it to the original interconnected system.
	\end{abstract}

	%%%%%%%%%%%Author Info%%%%%%%%%%%%%%%%%%%%%%%
	
	\title[Compositional abstractions of networks of stochastic hybrid systems under randomly switched topologies]{Compositional abstractions of networks of stochastic hybrid systems under randomly switched topologies}
	%\author[A. U. Awan]{Asad Ullah Awan} 
	%\author[M. Zamani]{Majid Zamani} 
	
	\author{Asad Ullah Awan}
	\author{Majid Zamani}
	\address{Department of Electrical and Computer Engineering, Technical University of Munich, D-80290 Munich, Germany.}
	\email{asad.awan@tum.de, zamani@tum.de}
\maketitle
%\begin{abstract}
%In this work, we derive conditions under which abstractions of networks of stochastic hybrid systems can be constructed compositionally. Proposed conditions leverage the interconnection topology, switching randomly between $\PC$ different interconnection topologies, and the joint dissipativity-type properties of subsystems and their abstractions. The random switching of the interconnection is modelled by a Markov chain. In the proposed framework, the abstraction, itself a stochastic hybrid system (possibly with a lower dimension), can be used as a substitute of the original system in the controller design process.  %In addition, we show that for a class of stochastic hybrid systems, the compositional abstraction can still be constructed if those conditions are violated for some of the $\PC$ interconnection topologies, provided that an additional condition on the Markov chain parameters is satisfied. 
%Finally, we provide an example illustrating the effectiveness of the proposed results by designing a controller enforcing some logic properties over the interconnected abstraction and then refining it to the original interconnected system.
%\end{abstract}

%%%%%%%%%%%%%%%%%%%%%%%%%%%%%%%%%%%%%%%%%%%%%
%%%%%Sec-Introduction%%%%%%%%%%%%%%%%%%%%%%%%
%%%%%%%%%%%%%%%%%%%%%%%%%%%%%%%%%%%%%%%%%%%%%
\section{Introduction}

For large interconnected control systems, controller design to achieve some complex specifications in a reliable and cost effective way is a challenging task. One direction which has been explored to overcome this challenge is to use a simpler (e.g. lower dimensional) (in)finite approximation of the given system as a replacement in the controller design process. This allows for an easier design of a controller for the approximation, which can be refined to the one for the original complex system via some interface. The error between the output behaviour of the concrete system and the one of its approximation can be quantified a priori in this detour controller synthesis approach. 

In addition, rather than treating the interconnected system in a monolithic manner, an approach which severely restricts the capability of existing techniques to deal with many number of subsystems, compositional schemes provide network-level certifications from main structural properties of the subsystems and their interconnections. In the past few years, there have been several results on the compositional construction of (in)finite abstractions of deterministic control systems including \cite{pola2016symbolic}, \cite{tazaki2008bisimilar}, \cite{RZ1}, and of a class of stochastic hybrid systems \cite{zamani2015approximations}. These results use small-gain type conditions to enable the compositional construction of abstractions. However, as shown in \cite{das2004some}, this type of conditions is a function of the size of the network and can be violated as the number of subsystems grows.

The recent result in \cite{7857702} provides a compositional framework for the construction of infinite abstractions of networks of control systems using dissipativity theory. In this result a notion of storage function is proposed which describes joint dissipativity properties of control systems and their abstractions. This notion is leveraged to derive compositional conditions under which a network of abstractions approximate a network of the concrete subsystems. Those conditions can be independent of the number of the subsystems under some properties on the interconnection topologies. This approach was recently extended to a class of stochastic hybrid systems in \cite{awan2017compositional}. 
%Stochastic hybrid systems provide a rich modelling framework for studying many complex systems including finance, biology, and bio-medicine \cite{oksendal2005applied}.

In more realistic scenarios, the interconnection topology of interconnected systems is not fixed due to for example loss of communication between the agents. To accommodate for this scenario, in this work, we deal with networks of stochastic hybrid systems wherein the interconnection topology is randomly switching between $\PC$ different interconnection matrices. We derive conditions under which one can construct an abstraction of a given network of stochastic hybrid systems under randomly switching topology in a compositional way. The random switching here is modeled using a continuous-time Markov chain, with each state of the chain representing a different interconnection topology. %We also consider a scenario wherein these conditions are not satisfied by some of the interconnection topologies. Inspired by a recent result in \cite{wang2017stability}, we show that compositional abstractions of a class of interconnected stochastic hybrid systems under randomly switching interconnection topologies can still be achieved provided an alternate weaker condition, and an additional condition on the Markov chain parameters. The condition on the Markov chain parameters can be interpreted as a probabilistic version of the so-called {\it dwell time} condition in deterministic switching systems \cite{liberzon2012switching}.
We illustrate the effectiveness of the approach by synthesizing a controller to enforce a given specification expressed as a linear temporal logic formula over the interconnected abstraction and then refining it to the original interconnected system.

% As in \cite{awanIFAC17}, the class of systems we consider is stochastic hybrid systems, which are a general class of systems consisting of continuous and discrete dynamics subject to probabilistic noise and events. In jump-diffusions, the continuous dynamics are modelled by stochastic differential equations and switches are modelled as Poisson processes, which provide a rich modelling framework for studying many complex systems including finance, biology, and bio-medicine \cite{oksendal2005applied}. %We illustrate the effectiveness of the proposed results by deriving conditions for a class of stochastic hybrid systems in which compositionality conditions are satisfied independent of the number or gains of the subsystems.
	
\section{Stochastic Hybrid Systems}
\subsection{Notation}
The sets of non-negative integer and real numbers are denoted by $\Natural$ and $\Real$, respectively. Those symbols are subscripted to restrict them in the usual way, e.g. $\Real_{>0}$ denotes the positive real numbers. The symbol $\Real^{n\times m}$ denotes the vector space of real matrices with $n$ rows and $m$ columns. The symbols $\vec{1}_{n}, \vec{0}_n, I_n, 0_{n \times m}$ denote the vector in $\Real^n$ with all its elements to be one, the zero vector, identity, and zero matrices in $\Real^n, \Real^{n \times n}, \Real^{n \times m}$, respectively. For $a,b \in \Real$ with $a \leq b$, the closed, open, and half-open intervals in $\Real$ are denoted by $[a,b], ]a,b[, [a,b[$, and $]a,b]$, respectively. For $a, b \in \Natural$ and $a \leq b$, we use $[a;b], ]a;b[, [a;b[$, and $]a;b]$ to denote the corresponding intervals in $\Natural$. Given $N \in \Natural_{\geq 1}$, vectors $x_i \in \Real^{n_i}, n_i \in \Natural_{\geq 1}$ and $i \in [1;N]$, we use $x = [x_1;\ldots;x_N]$ to denote the concatenated vector in $\Real^n$ with $n=\sum^N_{i=1} n_i$. Similarly, we use $X = [X_1;\ldots;X_N]$ to denote the matrix in $\Real^{n\times m}$ with $n = \sum^N_{i=1} n_i$, given $N \in \Natural_{\geq 1}$, matrices $X_i \in \Real^{n_i \times m}, n_i \in \Natural_{\geq 1}$, and $i \in [1;N]$. Given a vector $x \in \Real^{n}$, we denote by $\norm{x}$ the Euclidean norm of $x$. Given a matrix $M = \{m_{ij}\} \in \Real^{n\times m}$, we denote by $\norm{M}$ the Euclidean norm of $M$. %and by $\norm{M}_F$ the Frobenius norm of $M$, namely $\norm{M}_F = \sqrt{\mbox{Tr}(M^TM)}$, where $\mbox{Tr}(P) = \sum^n_{i=1}p_{ii}$ for any $P = \{p_{ij}\} \in \Real^{n\times n}$.
Given matrices $M_1,\dots,M_n$, the notation $\mathsf{diag}(M_1,\ldots,M_n)$ represents a block diagonal matrix with diagonal matrix entries $M_1,\ldots,M_n$. %Given a symmetric matrix $A$, $\lambda_{\min}(A)$ and $\lambda_{\max}(A)$ denote the minimum and maximum eigenvalues of $A$, respectively.
%Given a function $f$: $\Real^n \rightarrow \Real^m$ and $\bar{x} \in \Real^m$, we use $f \equiv \bar{x}$  to denote that $f(x) = \bar{x}$ for all $x \in \Real^n$. If $x$ is the zero vector, we simply write $f \equiv 0$. 
Given a function $f : \Real_{\geq 0} \rightarrow \Real^n$, the (essential) supremum of $f$ is denoted by $\norm{f}_{\infty} \coloneqq$  (ess)sup$\{\norm{f(t)}, \ t \geq 0\}$. A continuous function $\gamma: \Real_{\geq 0} \rightarrow \Real_{\geq 0}$, is said to belong to class $\mathcal{K}$ if it is strictly increasing and $\gamma(0) = 0$; $\gamma$ is said to belong to $\mathcal{K}_{\infty}$ if $\gamma \in \mathcal{K}$ and $\gamma(r) \rightarrow \infty$ as $r \rightarrow \infty$. A continuous function $\beta: \Real_{\geq 0} \times \Real_{\geq 0} \rightarrow \Real_{\geq 0}$ is said to belong to class $\mathcal{KL}$ if, for each fixed $t$, the map $\beta(r,t)$ belongs to class $\mathcal{K}$ with respect to $r$, and for each fixed non zero $r$, the map $\beta(r,t)$ is decreasing with respect to $t$ and $\beta(r,t) \rightarrow 0$ as $t \rightarrow \infty$. % $I(\cdot)$ denotes the indicator function i.e. $I(a) = 1$ if $a$ is true and $I(a) = 0$ otherwise.
%s 
% Given a set $U$ with subset $A$, the complement of $A$ with respect to $U$ is defined as $U \backslash A = \{x: x \in U , x \notin A\}.$ %\MZ{This complement definition does not make sense!} 	
\subsection{Stochastic hybrid systems}

Let $(\Omega, \Falgebra, \mathds{P})$ denote a probability space endowed with a filtration $\mathds{F} = (\mathcal{F}_s)_{s\geq0}$ satisfying the usual conditions of completeness and right continuity. The expected value of a measurable function $g(X)$, where $X$ is a random variable defined on a probability space ($\Omega, \Falgebra, \Pprob$), is defined by the Lebesgue integral $\Expectation{g(X)} \define \int_{\Omega} g(X(\omega))\mathds{P}(\diff\omega)$, where $\omega \in \Omega$. Let $(W_s)_{s\geq0}$ be a $\tilde{p}$-dimensional $\mathds{F}$-Brownian motion and $(P_s)_{s\geq 0}$ be a $\tilde{r}$-dimensional $\mathds{F}$-Poisson process. We assume that the Poisson process and Brownian motion are independent of each other. The Poisson process $P_s = [P_s^1;\ldots;P_s^{\tilde{r}} ]$ models $\tilde{r}$ kinds of events whose occurrences are assumed to be independent of each other. 
%%%%%%%%%%%%%%%%%%%%%%%%%%%%%%%%%%%%%%%%%%%%%%%%%%%%%%%%%%%%%%%%%%%%%%%%%%%%%%%%%%%%%%%%%%
%%%%%%%%%%%%%%Definition-  Stochastic hybrid systems%%%%%%%%%%%%%%%%%%%%%%%%%%%%%%%%%%%%%%
%%%%%%%%%%%%%%%%%%%%%%%%%%%%%%%%%%%%%%%%%%%%%%%%%%%%%%%%%%%%%%%%%%%%%%%%%%%%%%%%%%%%%%%%%%
\begin{definition}  
	\label{def:shs}
	The class of stochastic hybrid systems studied in this paper is a tuple  $$\Sigma = (\Real^n, \Real^m, \Real^p, \mathcal{U}, \mathcal{W}, f, \sigma, \rho,\Real^{q_1}, \Real^{q_2},h_1,h_2),$$ where
	\begin{itemize}
		
		\item $\Real^n$, $\Real^m$, $\Real^p$, $\Real^{q_1}$, and $\Real^{q_2}$ are the state, external input, internal input, external output, and internal output spaces, respectively;
		\item $\ExtInputMap$ and $\IntInputMap$ are subsets of sets of all $\mathds{F}$-progressively measurable processes with values in $\Real^m$ and $\Real^p$, respectively;
		\item $f:\Real^n \times \Real^m \times \Real^p \rightarrow \Real^n $  is the drift term which is globally Lipschitz continuous: there exists Lipschitz constants $L_x, L_u, L_w \in \Real_{\geq 0}$ such that $\norm{f(x,u,w) - f(x',u',w')} \leq L_x\norm{x-x'} + L_u\norm{u - u'} + L_w\norm{w-w'}$ for all $x,x'  \in \Real^n$, all $u,u' \in \Real^m$, and all $w, w'\in \Real^p$; 
		\item $\sigma: \Real^n \rightarrow \Real^{n\times \tilde{p}}$ is the diffusion term which is globally Lipschitz continuous with the Lipschitz constant $L_{\sigma}$;
		\item $\rho: \Real^n \rightarrow \Real^{n \times \tilde{r}}$ is the reset term which is globally Lipschitz continuous with the Lipschitz constant $L_{\rho}$;
		\item $h_1: \Real^n \rightarrow \Real^{q_1}$ is the external output map;
		\item $h_2: \Real^n \rightarrow \Real^{q_2}$ is the internal output map.
		%\item  is a continous function of its argument satisfying 
		%\item $\mathbb{U} \subseteq \mathbb{R}^n$ is the input set \\
		%\item $\mathcal{U} is a subset of the set of all measurebale, locally essetnaiyll bounded 
	\end{itemize}
\end{definition}
A stochastic hybrid system $\Sigma$ satisfies
\begin{IEEEeqnarray}{c}	
	\label{eq:w}
	\Sigma:\left\{ \begin{IEEEeqnarraybox}[\relax][c]{rCl}
		\diff \xi(t) &=& f(\xi(t), \upsilon(t), \omega(t))\diff t + \sigma(\xi(t))\diff W_t   \\ && +   \rho(\xi(t)) \diff P_t, \\
		\zeta_1(t) &=& h_1(\xi(t)), \\
		\zeta_2(t) &=& h_2(\xi(t)),
		\end{IEEEeqnarraybox}\right.
\end{IEEEeqnarray}
$\Pprob$-almost surely ($\Pprob$-a.s.) for any $\upsilon \in \ExtInputMap$ and any $\omega \in \IntInputMap$, where stochastic process $\xi: \Omega \times \Real_{\geq 0} \rightarrow \Real^n$ is called a {\it solution process} of $\Sigma$, stochastic process $\zeta_1 : \Omega \times \Real_{\geq 0} \rightarrow \Real^{q_1}$ is called an external output trajectory of $\Sigma$, and  stochastic process $\zeta_2: \Omega \times \Real_{\geq 0} \rightarrow \Real^{q_2}$ is called an internal output trajectory of $\Sigma$. 
%We call the tuple $(\xi, \zeta_1, \zeta_2, \upsilon, \omega)$ a {\it trajectory} of $\Sigma$, consisting of a solution process $\xi$, output trajectories $\zeta_1$ and $\zeta_2$, and input trajectories $\upsilon$ and $\omega$, that satisfies \eqref{eq:w}. 
We also write $\xi_{a\upsilon\omega}(t)$ to denote the value of the solution process at time $t \in \Real_{\geq 0}$ under the input trajectories $\upsilon$ and $\omega$ from initial condition $\xi_{a\upsilon\omega}(0) = a$ $\Pprob$-a.s., where $a$ is a random variable that is $\mathcal{F}_0$-measurable. We denote by $\zeta_{1_{a\upsilon\omega}}$ and $\zeta_{2_{a\upsilon\omega}}$ the external and internal output trajectories corresponding to the solution process $\xi_{a\upsilon\omega}$. Here, we assume that the Poisson processes $P^i_s$, for any $i \in [1;\tilde{r}]$, have the rates $\lambda_i$. We emphasize that the postulated assumptions on $f, \sigma$, and $\rho$ ensure existence, uniqueness, and strong Markov property of the solution process \cite{oksendal2005applied}. %Citation over here.

We now introduce the definition of a continuous-time Markov chain, used later to model the switching between interconnection topologies.
\begin{definition}
	\label{def:mc}
	A continuous-time Markov chain is a tuple $\Pi = (\modes, \mathsf{Q})$, 
	where 
	\begin{itemize}
		\item $\modes$ is a finite set of cardinality $\PC$, called the state-space of the Markov-chain;
		\item $\mathsf{Q} = \{q_{ij}\} \in \Real^{\PC \times \PC }$ is called the generator matrix.
	\end{itemize}
Associated with $\Pi$ is a stochastic process $\widehat{\pi}: \Omega \times \Real_{\geq 0} \rightarrow \modes$, on the probability space  $(\Omega, \Falgebra, \mathds{P})$, such that for every fixed $\omega \in \Omega$, $\pi(\cdot) = \widehat{\pi}(\omega,\cdot): \Real_{\geq 0} \rightarrow \modes$. For any $i,j \in \modes$ and $t \in \Real_{\geq 0},$ one has
	\begin{align}
	\mathds{P}(\pi(t + h) = j \vert \pi(t) = i) = \begin{cases}
	q_{ij}h + o(h), \ \ \ \ \ \ i \neq j, \\
	1 + q_{ii}h + o(h), \ i = j,
	\end{cases} 
	\end{align} 
	where $h > 0$, $\lim_{h\to \infty}\frac{o(h)}{h} = 0$, $q_{ii} = -\sum_{i\neq j} q_{ij}$, and $q_{ij} \geq 0$ is called the transition jump rate from $i$ to $j$ if $i \neq j$. We denote the value of the solution process at time $t \in \Real_{\geq 0}$ by $\pi(t)$, and refer to it as the {\it switching process}.% We assume that $\Pi$ is ergodic. This implies that there exists a unique stationary distribution denoted by $\pi_s = (\pi_{s1}, \dots, \pi_{s\PC})$, such that: 
%\begin{align}
%	\lim_{t \rightarrow \infty}\mathds{P}(\pi(t + h) = j \vert \pi(t) = i) = \pi_{sj}, \quad \forall i \in \modes.
%\end{align} 
%We also assume that $pi(0)$ is deterministic. 
\end{definition}
We now introduce the definition of a switching stochastic hybrid system, which will be used later to denote interconnected stochastic hybrid systems with randomly switching topologies. 
%%%Switchign stochastic hybrid system%%%%%%%%%%%%%%%%%%%%%%%%%
\begin{definition}
	A switching stochastic hybrid system is a tuple  $\Sigma = (\Real^n, \Real^m, \mathcal{U}, \modes,\mathscr{P}, f, \sigma, \rho, \Real^q, h)$ where
	%$\Sigma = (\Real^n, \Real^m, \mathcal{U}, \modes, \mathscr{P}, f, \sigma, \rho, \Real^q, h)$, where
		\begin{itemize}
		
		\item $\Real^n$, $\Real^m$, and $\Real^{q}$,  are the state, external input and external output spaces, respectively;
		\item $\ExtInputMap$ is a subset of the set of all $\mathds{F}$-progressively measurable processes with values in $\Real^m$;
		\item $\modes = \{1,\dots,
		\PC \}$ is a finite set of modes;
		\item $\mathscr{P}$ is a subset of the set of all piecewise constant c$\grave{\text{a}}$dl$\grave{\text{a}}$g (i.e. right continuous and with left limits) functions of time from $\Real_{\geq 0}$ to $\modes$ and characterized by a finite number of discontinuities on every bounded interval in $\Real_{\geq 0}$ (no Zeno behaviour);
		%\item $\Pi = (\modes,Q)$ is a continuous-time Markov chain, as in Definition \ref{def:mc}, called the {\it switching process}
		\item $f : \Real^n \times \Real^m \times 
		\modes \rightarrow \Real^n$, is a drift term such that for every fixed $\mathsf{p} \in \modes$, $f(.,.,\mathsf{p})$ is globally Lipschitz continuous: there exists Lipschitz constants $L_x, L_u \in \Real_{\geq 0}$ such that $\norm{f(x,u,\mathsf{p}) - f(x',u',\mathsf{p})} \leq L_x\norm{x-x'} + L_u\norm{u - u'}$ for all $x,x'  \in \Real^n$ and all $u,u' \in \Real^m$; 
		\item $\sigma: \Real^n \rightarrow \Real^{n\times \tilde{p}}$ is the diffusion term which is globally Lipschitz continuous with the Lipschitz constant $L_{\sigma}$;
		\item $\rho: \Real^n \rightarrow \Real^{n \times \tilde{r}}$ is the reset term which is globally Lipschitz continuous with the Lipschitz constant $L_{\rho}$;
		\item $h: \Real^n \rightarrow \Real^{q}$ is the external output map;
		\end{itemize}
	%	Let $\Pi = (\modes, \mathsf{Q})$ be a continuous-time Markov chain as in Definition \ref{def:mc} with switching process $\pi$. Given a switching process $\pi$, 
	The stochastic process $\xi: \Omega \times \Real_{\geq 0} \rightarrow \Real^n$ is a solution process of $\Sigma$ if there exists $\upsilon \in \ExtInputMap$ and $\pi \in \mathscr{P}$ satisfying
%		\begin{align}
%		\label{eq:sde}
%		\Sigma: \begin{cases} 
%		\diff \xi(t) = f(\xi(t), \upsilon(t), \pi(t))\diff t + \sigma(\xi(t))\diff W_t \\ +   \rho(\xi(t)) \diff P_t \\
%		\zeta(t) = h(\xi(t)),
%		\end{cases}
%		\end{align}
		\begin{IEEEeqnarray}{c}	
		\label{eq:sdeSW}
			\Sigma:\left\{ \begin{IEEEeqnarraybox}[\relax][c]{rCl}
			\diff \xi(t) &=& f(\xi(t), \upsilon(t), \pi(t))\diff t + \sigma(\xi(t))\diff W_t \nonumber \\ && +   \rho(\xi(t)) \diff P_t \\
			\zeta(t) &=& h(\xi(t)),
			\end{IEEEeqnarraybox}\right.
		\end{IEEEeqnarray}
		$\Pprob$-almost surely ($\Pprob$-a.s.) at each time $t \in \Real_{\geq 0}$. We denote by $\xi_{a\upsilon}(t, \pi(t))$ the value of the solution process at time $t \in \Real_{\geq 0}$ under the control input $\upsilon \in \ExtInputMap$ and the switching process $\pi \in \mathscr{P}$, starting from an  initial condition $\xi_{a\upsilon}(0, \pi(0)) = a$, $\Pprob$-a.s., where $a$ is a random variable that is $\mathcal{F}_0$-measurable.
		%\item  is a continous function of its argument satisfying 
		%\item $\mathbb{U} \subseteq \mathbb{R}^n$ is the input set \\
		%\item $\mathcal{U} is a subset of the set of all measurebale, locally essetnaiyll bounded 
\end{definition}

%%%%%%%%%%%%
%%%%%%%%%%%%%%%%%%%%%%%%%%%%%%%%%%%%%%%%%%%%%%%%%%%%%%%%%%%%%%%%%%%%%%%%%%%%%%%%%%%%%%%%%%
%%%%%%%%%%%%%%Section-  Stochastic storage function%%%%%%%%%%%%%%%%%%%%%%%%%%%%%%%%%%%%%%%
%%%%%%%%%%%%%%%%%%%%%%%%%%%%%%%%%%%%%%%%%%%%%%%%%%%%%%%%%%%%%%%%%%%%%%%%%%%%%%%%%%%%%%%%%%
\section{Stochastic Storage Function}
In this section, we introduce a notion of stochastic storage function, adapted from the notion of storage functions from dissipativity theory \cite{arcak2016networks}.
\begin{definition}
	Let $$\Sigma = (\Real^n, \Real^m, \Real^p, \mathcal{U}, \mathcal{W}, f, \sigma, \rho, \Real^{q_1}, \Real^{q_2}, h_1,h_2),$$ and $$\hat{\Sigma} = (\Real^{\hat{n}}, \Real^{\hat{m}}, \Real^{\hat{p}}, \mathcal{\hat{U}}, \mathcal{\hat{W}}, \hat{f}, \hat{\sigma}, \hat{\rho}, \Real^{q_1}, \Real^{\hat{q}_2}, \hat{h}_1,\hat{h}_2),$$ be two stochastic hybrid systems with the same external output space dimension, and with solution processes $\xi$ and $\hat{\xi}$, respectively. A twice continuously differentiable function $\mathcal{S}:\Real^n \times \Real^{\hat{n}} \rightarrow \Real_{\geq 0}$ is called a stochastic storage function from $\hat{\Sigma}$ to $\Sigma$ in the second moment (SStF-M$_2$), if it has polynomial growth rate and if there exists convex function $\alpha \in \mathcal{K}_{\infty}$, concave function $\psi_{\mathrm{ext}} \in \mathcal{K}_{\infty} \cup \{0\} $, some positive constant $\kappa$, some matrices $W, \hat{W},$ and $H$ of appropriate dimensions, and some symmetric matrix $X$ of appropriate dimension with conformal block partitions $X^{ij}, i,j \in [1;2]$, such that $\forall x \in \Real^n$ and $\forall\hat{x} \in \Real^{\hat{n}}$, one has 
	\begin{align}
	\label{in:defV}
	\alpha(\norm{h_1(x) - \hat{h}_1(\hat{x})}^2) \leq \mathcal{S}(x,\hat{x}),
	\end{align}
	and $\forall x \in \Real^n$ $\forall \hat{x} \in \Real^{\hat{n}}$ $\forall \hat{u} \in \Real^{\hat{m}} \ \exists u \in \Real^m$ $\forall \hat{w} \in \Real^{\hat{p}} \ \forall w \in \Real^p$, one obtains
	\begin{align}
	\label{ineq:defDiss}
	& \infgen \mathcal{S}(x,\hat{x}) \leq-\kappa \mathcal{S}(x,\hat{x}) + \psi_{\mathrm{ext}}(\norm{\hat{u}}^2) \nonumber \\ & +  \begin{bmatrix} Ww - \hat{W}\hat{w} \\ h_2(x) - H\hat{h}_2(\hat{x})\end{bmatrix}^T \begin{bmatrix} X^{11} & X^{12} \\ X^{21} & X^{22}\end{bmatrix} \begin{bmatrix} Ww - \hat{W}\hat{w} \\ h_2(x) - H\hat{h}_2(\hat{x})\end{bmatrix},
	\end{align}
	\normalsize
	where $\infgen$ denotes the infinitesimal generator of the stochastic process $\Xi = [\xi;\hat{\xi}]$, acting on the function $\mathcal{S}$ (see \cite{oksendal2005applied} for a definition).
%	We use notation $\hat{\Sigma} \preceq \Sigma$ if there exists a SStF-M$_2$ $\mathcal{S}$ from $\hat{\Sigma}$ to $\Sigma$. 
	%where $2\underbar{k} = k$.
\end{definition}
The stochastic hybrid system $\hat{\Sigma}$ (possibly with $\hat{n} < n$) is called an abstraction of $\Sigma$. 
%%%%%%%%%%Markovian switching signal%%%%%%%%%%%%%%%%%%%%%%%%%%%%%%%

%%%%%%%%%%%%%%%%%%%%%%%%%%%%%%%%%%%%%%%%%%%%%%%%%%%%%%%%%%%%%%%%%%%%%%%%%%%%%%%%%%%%%%%%%%
%%%%%%%%%%%%%%Definition -  Switching Stochastic Simulation Function %%%%%%%%%%%%%%%%%%%%%%%%%%%%%%%
%%%%%%%%%%%%%%%%%%%%%%%%%%%%%%%%%%%%%%%%%%%%%%%%%%%%%%%%%%%%%%%%%%%%%%%%%%%%%%%%%%%%%%%%%%
%\begin{remark}
%	If the stochastic hybrid system $\Sigma$ does not have internal inputs and outputs, the definition of the system defined in Definition \ref{def:shs} reduces to tuple $\Sigma = (\Real^n, \Real^m, \mathcal{U}, f, \sigma, \rho, \Real^{q}, h)$ where $f:\Real^n \times \Real^m \rightarrow \Real^n$. Correspondingly, the equation \eqref{eq:sde} describing the evolution of solution processes reduces to:
%	\begin{align}
%	\label{eq:reduced_sde}
%	\Sigma: \begin{cases} 
%	\diff \xi(t) = f(\xi(t), \upsilon(t))\diff t + \sigma(\xi(t))\diff W_t \\  \qquad\qquad +   \rho(\xi(t)) \diff P_t, \\
%	\zeta(t) = h(\xi(t)). \\
%	\end{cases}
%	\end{align}
%\end{remark}
%Now, we recall the notion of stochastic simulation function function introduced in \cite{Zamani2016}. 
\begin{definition}
	Let $\Pi = (\modes, \mathsf{Q})$ be a continuous-time Markov chain with switching process $\pi$. Let  $\Sigma = (\Real^n, \Real^m, \mathcal{U}, \modes, \pi, f, \sigma, \rho, \Real^q, h)$ and $\hat{\Sigma} = (\Real^{\hat{n}}, \Real^{\hat{m}}, \mathcal{\hat{U}}, \modes, \pi, \hat{f}, \hat{\sigma}, \hat{\rho}, \Real^{q}, \hat{h}$) be two switching stochastic hybrid systems.  % t \in \Real_{\geq 0}$.
	 A twice continuously differentiable function $V: \Real^n \times \Real^{\hat{n}} \times \modes \rightarrow \Real_{\geq 0}$ is called a stochastic simulation function in the second moment  (SSF-M$_2$), from $\hat{\Sigma}$ to $\Sigma$ if there exists convex function $\alpha \in \mathcal{K}_{\infty}$, concave function $\psi_{\mathrm{ext}} \in \mathcal{K}_{\infty} \cup \{0\}$, and positive constant $\kappa$, such that for any $x \in \Real^n$, any $\hat{x} \in \Real^{\hat{n}}$, and any $j \in \modes$, one has
	\begin{align}
	\label{ineq:simfunction1}
	\alpha(\norm{h(x) - \hat{h}(\hat{x})}^2) \leq V(x,\hat{x},j),
	\end{align}
	and $\forall j \in \modes$ $\forall x \in \Real^n$ $\forall \hat{x} \in \Real^{\hat{n}}$ $\forall \hat{u} \in \Real^{\hat{m}} \ \exists u \in \Real^m$ such that
	\begin{align}
	\label{ineq:simfunction2}
	\infgen V(x,\hat{x},j) \leq -\kappa V(x,\hat{x},j) + \psi_{\mathrm{ext}}(\norm{\hat{u}}^2).
	\end{align}
\end{definition}
We say that a switching stochastic hybrid  system $\hat{\Sigma}$ is approximately simulated by a switching stochastic hybrid system $\Sigma$ %, denoted by $\hat{\Sigma} \preceq_{\mathcal{AS}} \Sigma$, 
if there exists an SSF-M$_2$ $V$ from $\hat{\Sigma}$ to $\Sigma$. We call $\hat{\Sigma}$ (possibly with $\hat n < n$) an abstraction of $\Sigma$.

The next theorem shows the important of the existence of an SSF-M$_2$ by quantifying the error between the output behaviours of $\Sigma$ and those of its abstraction $\hat{\Sigma}$. 
%%%%%%%%%%%%%%%%%%%%%%%%%%%%%%%%%%%%%%%%%%%%%%%%%%%%%%%%%%%%%%%%%%%%%%%%%%%%%%%%%%%%%%%%%%
%%%%%%%%%%%%%%Theorem -  Simulation Function - Closeness of trajectories%%%%%%%%%%%%%%%%%%
%%%%%%%%%%%%%%%%%%%%%%%%%%%%%%%%%%%%%%%%%%%%%%%%%%%%%%%%%%%%%%%%%%%%%%%%%%%%%%%%%%%%%%%%%%
\begin{theorem}
 Let $\Pi = (\modes, \mathsf{Q})$ be a continuous-time Markov chain with switching process $\pi$. Let us consider two switching stochastic hybrid systems $\Sigma = (\Real^n, \Real^m, \mathcal{U}, \modes, \pi,f, \sigma, \rho, \Real^q, h)$ and $\hat{\Sigma} = (\Real^{\hat{n}}, \Real^{\hat{m}}, \mathcal{\hat{U}}, \modes,\pi, \hat{f}, \hat{\sigma}, \hat{\rho}, \Real^{q}, \hat{h})$. % t \in \Real_{\geq 0}$. 
  Suppose $V$ is an SSF-M$_2$ from $\hat{\Sigma}$ to $\Sigma$. Then, there exists a $\mathcal{KL}$ function $\beta$ and a function $\gamma_{\mathrm{ext}}\in\mathcal{K}_{\infty} \cup \{0\}$ such that for any random variables $a$ and $\hat{a}$ that are $\mathcal{F}_0$-measurable, and for any $\hat{\upsilon}\in\hat{\ExtInputMap}$  there exists $\upsilon \in \ExtInputMap$ such that the following inequality holds for any $t \in \Real_{\geq 0}$:
	\begin{align}
	\label{eq:bound_output}
	\Expectation{\norm{\zeta_{a\upsilon}(t) - \hat{\zeta}_{\hat{a}\hat{\upsilon}}(t)}^2} &\leq \beta(\Expectation{V(a,\hat{a},\pi(0))},t)\nonumber\\&\quad + \gamma_{\mathrm{ext}}(\Expectation{\norm{\hat{\upsilon}}^2_{\infty}}).
	\end{align}
\end{theorem}
%%%%%%%%%%%%%%%%%%%%%%%%%%%%%%%%%%%%%%%%%%%%%%%%%%%%%%%%%%%%%%%%%%%%%%%%%%%%%%%%%%%%%%%%%%
%%%%%%%%%%%%%%Proof- Simulation Function - Closeness of trajectories%%%%%%%%%%%%%%%%%%%%%%
%%%%%%%%%%%%%%%%%%%%%%%%%%%%%%%%%%%%%%%%%%%%%%%%%%%%%%%%%%%%%%%%%%%%%%%%%%%%%%%%%%%%%%%%%%
\begin{proof}
	The proof is similar to that of Theorem 3.29 in \cite{chatterjee2007studies} and is omitted. 
\end{proof}
%%%%%%%%%%%%%%%%%%%%%%%%%%%%%%%%%%%%%%%%%%%%%%%%%%%%%%%%%%%<Too big to fit in double column
\section{Compositionality Result}
First, we introduce interconnected stochastic hybrid systems.
\begin{definition}
Consider $N \in \Natural_{\geq1}$ stochastic hybrid subsystems $$\Sigma_i = (\Real^{n_i}, \Real^{m_i}, \Real^{p_i}, \mathcal{U}_i, \mathcal{W}_i, f_i, \sigma_i, \rho_i, \Real^{q_{1i}}, \Real^{q_{2i}}, h_{1i},h_{2i}),$$ where $i \in [1;N]$. Consider a continuous-time Markov chain $\Pi= (\modes,\mathsf{Q})$, as in Definition \ref{def:mc}, with $\modes = \{1,\dots,\PC\}$, and switching process $\pi$. Consider a set of interconnection matrices $M = \{M_1,\dots, M_{\PC}\}$, where each matrix $M_i$, $i \in [1;\PC]$, defines the coupling of these subsystems. The interconnected switching stochastic hybrid system $$\Sigma = (\Real^{n}, \Real^{m},\mathcal{U}, \modes, \pi, f, \sigma, \rho, \Real^{q}, h),$$ denoted by $\mathcal{I}_{\pi}^M(\Sigma_1,\dots,\Sigma_N)$, follows by $n = \sum^N_{i=1} n_i, m = \sum^{N}_{i=1} m_i, q = \sum^{N}_{i=1}q_{1i}$, and the functions 
\begin{align}
f(x,u,\pi(t)) &\define [f_1 (x_1,u_1,w_1);\dots;f_N (x_N,u_N,w_N)], \\
\sigma(x) &\define [\sigma_1(x_1);\dots;\sigma_N(x_N)], \\
\rho(x) &\define [\rho_1(x_1);\dots;\rho_N(x_N)], \\
h(x) &\define [h_{11}(x_1);\dots;h_{1N} (x_N)],
\end{align}
where $u = [u_1;\dots;u_N]$, $x=[x_1;\dots;x_N]$ and at any time $t \in \Real_{\geq 0}$, the internal variables are constrained by 
\begin{align}
[w_1;\dots;w_N] = M_{\pi(t)}[h_{21}(x_1);\dots;h_{2N}(x_N)].
\end{align} 
\end{definition}
The next theorem provides a compositional approach on the construction of abstractions of networks of stochastic hybrid systems under randomly switching interconnection topology.

\begin{theorem}
Consider an interconnected switching stochastic hybrid system $\Sigma = \mathcal{I}_{
\pi}^M(\Sigma_1,\dots,\Sigma_N)$ induced by $N \in \Natural_{\geq 1}$ stochastic hybrid subsystems $\Sigma_i$, a set of the interconnection matrices  $M = \{M_1, \dots, M_{\PC}\}$, and a continuous-time Markov chain $\Pi = (\modes, \mathsf{Q})$ governing the switching of the interconnection topology with associated stochastic process $\pi$. More specifically, the interconnection topology at any time $t \in \Real_{\geq 0}$ is $M_{\pi(t)}$. Suppose each subsystem $\Sigma_i$ admits an abstraction $\hat{\Sigma}_i$ with the corresponding SStF-M$_2$ $\mathcal{S}_i$. If there exists a finite set of matrices $\hat{M} = \{\hat{M}_1,\dots,\hat{M}_{\PC}\}$ of appropriate dimension such that for each $j \in [1;\PC]$ the matrix (in)equalities
\begin{align} 
\label{eq:interconnected1}
\begin{bmatrix}WM_j \\ I_{\tilde{q}} \end{bmatrix}^T X(\mu_1^jX_1,\dots,\mu_N^j X_N) \begin{bmatrix}WM_j \\ I_{\tilde{q}} \end{bmatrix}    \preceq 0,  
\end{align} 
\begin{align} 
\label{eq:interconnected2}
WM_jH = \hat{W}\hat{M}_j,
\end{align}
are satisfied for some $\mu_i^j > 0$, $i \in [1;N]$, where  $\tilde{q} = \sum^N_{i=1}q_{2i}$ and
\begin{align}
W &= \mathsf{diag}(W_1,\ldots,W_N), 
\hat{W} = \mathsf{diag}(\hat{W}_1,\ldots,\hat{W}_N), \nonumber\\
H &= \mathsf{diag}(H_1,\ldots,H_N),
%\begin{bmatrix}
%H_1 & & \\
%& \ddots & \\
%& & H_N
%\end{bmatrix}, 
\end{align}
\begin{align}
\label{eq:X}
&X(\mu_1^j X_1,\dots,\mu_N^jX_N) =\nonumber \\
&\begin{bmatrix}
\mu_1^jX_1^{11} & & & \mu_1^j X_1^{12} & & \\
& \ddots & & & \ddots & \\
& & \mu_N^j X_N^{11} & & & \mu_N^jX_N^{12} \\
\mu_1^jX_1^{21} & & & \mu_1^j X_1^{22} & & \\
& \ddots & & & \ddots & \\
& & \mu_N^j X_N^{21} & & & \mu_N^jX_N^{22}
\end{bmatrix},
\end{align}
\normalsize
then $$V(x,\hat{x},j) \define \sum^{N}_{i=1}\mu_i^j \mathcal{S}_i(x_i,\hat{x}_i),$$ is an SSF-M$_2$ from the interconnected switching stochastic hybrid system $\hat{\Sigma} \define \mathcal{I}_{\pi}^{\hat{M}}(\hat{\Sigma_1},\dots,\hat{\Sigma}_N)$, with the interconnection matrix at time $t \in \Real_{\geq 0}$ given by $\hat{M}_{\pi(t)}$, to $\Sigma$.
\end{theorem}

\begin{proof}
%%%%%%%%%%%%%%%%%%%%%%%%%%%%%%%%%%%%%%%%%%%%%%%%%%%%%%%%%%<Too big to fit in double column
The proof is inspired by that of Theorem 4.2 in \cite{7857702}. First we show that inequality (\ref{ineq:simfunction1}) holds for some convex $\mathcal{K}_{\infty}$ function $\alpha$. For any $x = [x_1;\dots;x_N] \in \Real^n$, $\hat{x} = [\hat{x}_1;\dots;\hat{x}_N] \in \Real^{\hat{n}}$, and $j\in \modes$, one gets:
\begin{align*}
\norm{h(x) - \hat{h}(\hat{x})}^2 &\leq  \sum^{N}_{i=1}\norm{h_{1i}(x_i) - \hat{h}_{1i}(\hat{x}_i)}^2 \nonumber \\ 
&\leq\sum^{N}_{i=1}\alpha^{-1}_i(\mathcal{S}_i(x_i,\hat{x}_i)) \leq \underline{\alpha}_j(V(x,\hat{x},j)), \nonumber
\end{align*}
where $\underline{\alpha}_j$ is a $\mathcal{K}_{\infty}$ function defined as 
\begin{align}
\underline{\alpha}_j(s) \define \begin{cases}   \max\limits_{\vec{s}\geq 0} & \sum^N_{i=1}\alpha^{-1}_i(s_i)\\ \mbox{s.t.} & \mu_j^T\vec{s} = s,\end{cases}
\end{align}
where $\vec{s} = [s_1;\dots;s_N] \in \Real^N$ and $\mu_j = [\mu_1^j;\dots;\mu_N^j]$. Since $\underline{\alpha}_j\in\mathcal{K}_\infty$ are concave functions as argued in \cite{zamani2015approximations}, there exists a concave function $\underline{\alpha}\in\mathcal{K}_\infty$ such that $\underline{\alpha}_j\leq\underline{\alpha}$ $\forall j\in \modes$.
\begin{figure*}
%\horizontalEqSep
\small
\begin{align}
\label{ineq:inf_storage}
&\infgen V(x,\hat{x},j) = \sum^N_{i=1}\mu_i^j\infgen \mathcal{S}_i(x_i,\hat{x}_i) \nonumber
\\ &\leq \sum^N_{i=1} \mu_i^j \Bigg(-\kappa_i \mathcal{S}_i(x_i,\hat{x}_i)	 + \psi_{i\mathrm{ext}}(\norm{\hat{u}_i}^2) \nonumber  +  \begin{bmatrix} W_i w_i - \hat{W}_i\hat{w_i} \\ h_{2i}(x_i) - H_i\hat{h}_{2i}(\hat{x_i})\end{bmatrix}^T \begin{bmatrix} X^{11}_i & X^{12}_i \\ X^{21}_i & X^{22}_i\end{bmatrix} \begin{bmatrix} W_i w_i - \hat{W}_i\hat{w}_i \\ h_{2i}(x_i) - H_i\hat{h}_{2i}(\hat{x}_i)\end{bmatrix}\Bigg)\nonumber  \\
&\leq  \sum^N_{i=1} -\mu_i^j \kappa_i \mathcal{S}_i(x_i,\hat{x}_i) + \sum^N_{i=1}\mu_i^j\psi_{i\mathrm{ext}}(\norm{\hat{u}_i}^2)\nonumber +  \begin{bmatrix} W\begin{bmatrix}w_1 \\ \vdots \\ w_N \end{bmatrix} - \hat{W}\begin{bmatrix}\hat{w}_1 \\ \vdots \\ \hat{w}_N \end{bmatrix} \\ h_{1}(x_1) - H_1\hat{h}_{21}(\hat{x}_1) \\ \vdots \\ h_{2N}(x_N) - H_N\hat{h}_{2N}(\hat{x}_N) \end{bmatrix}^T X(\mu_1^jX_1,\dots,\mu_N^j X_N)  \begin{bmatrix} W\begin{bmatrix}w_1 \\ \vdots \\ w_N \end{bmatrix} - \hat{W}\begin{bmatrix}\hat{w}_1 \\ \vdots \\ \hat{w}_N \end{bmatrix} \\ h_{1}(x_1) - H_1\hat{h}_{21}(\hat{x}_1) \\ \vdots \\ h_{2N}(x_N) - H_N\hat{h}_{2N}(\hat{x}_N) \end{bmatrix} \nonumber\\
&\leq  \sum^N_{i=1} -\mu_i^j \kappa_i\mathcal{S}_i(x_i,\hat{x}_i) + \sum^N_{i=1}\mu_i^j\psi_{i\mathrm{ext}}(\norm{\hat{u}_i}^2)\nonumber \\
& \quad + \begin{bmatrix} h_{21}(x_1) - H_1\hat{h}_{21}(\hat{x}_1) \\ \vdots \\ h_{2N}(x_N) - H_N\hat{h}_{2N}(\hat{x}_N) \end{bmatrix}^T \begin{bmatrix}WM_j \\ I_{\tilde{q}} \end{bmatrix}^T X(\mu_1^jX_1,\dots,\mu_N^j X_N) \begin{bmatrix}WM_j \\ I_{\tilde{q}} \end{bmatrix}  \begin{bmatrix} h_{21}(x_1) - H_1\hat{h}_{21}(\hat{x}_1) \\ \vdots \\ h_{2N}(x_N) - H_N\hat{h}_{2N}(\hat{x}_N) \end{bmatrix} \nonumber\\
&\leq -\kappa V(x,\hat{x},j) + \psi_{\mathrm{ext}}(\norm{\hat{u}}^2).
\end{align} \nonumber 
\normalsize
\horizontalEqSep
\vspace{-1cm}
\end{figure*}
By defining $\alpha = \underline{\alpha}^{-1}$ which is a convex $\mathcal{K}_\infty$ function, one obtains
$$\alpha(\norm{h_{1}(x) - \hat{h}_1(\hat{x})}^2) \leq V(x,\hat{x},j),$$
satisfying inequality \eqref{ineq:simfunction1}. 
Now we show inequality \eqref{ineq:simfunction2}. Consider any $x=[x_1;\dots;x_N] \in \Real^n, \hat{x} = [\hat{x}_1;\dots;\hat{x}_N] \in \Real^{\hat{n}},$ $\hat{u} = [\hat{u}_1;\dots;\hat{u}_N] \in \Real^{\hat{m}}$ and $j \in \modes$. For any $i \in [1;N]$, there exists $u_i \in \Real^{m_i}$, consequently, a vector $u = [u_1;\ldots;u_N] \in \Real^m$, satisfying  \eqref{ineq:defDiss} for each pair of subsystems $\Sigma_i$ and $\hat{\Sigma}_i$ with the internal inputs given by $[w_1;\dots;w_N] = M_j[h_{21}(x_1);\dots;h_{2N}(x_N)]$ and $[\hat{w}_1;\dots;\hat{w}_N] = \hat{M}_j[\hat{h}_{21}(\hat{x}_1);\dots;\hat{h}_{2N}(\hat{x}_N)]$. We consider the infinitesimal generator of function $V$ and employ conditions \eqref{eq:interconnected1} and \eqref{eq:interconnected2} which result in the chain of inequalities \eqref{ineq:inf_storage}. In \eqref{ineq:inf_storage} the constant $\kappa = \min\limits_{i \in [1;N]} \kappa_i$ and the function $\psi_{\mathrm{ext}} \in \mathcal{K}_{\infty} \cup \{0\}$ is defined as the following. Consider $\mathcal{K}_\infty\cup \{0\}$ functions
\normalsize
%$$\kappa \coloneqq \begin{cases} \min_{\vec{s} \geq 0} & \sum^N_{i=1} \mu_i^j\eta_i(s_i) \\ \mbox{s.t.} & \mu^T\vec{s} = s,\end{cases}$$
$$\psi_{\mathrm{ext}}^j(s) \coloneqq \begin{cases} \max\limits_{\vec{s} \geq 0} & \sum^N_{i=1} \mu_i^j\psi_{i\mathrm{ext}}(s_i) \\ \mbox{s.t.} & \norm{\vec{s}} \leq s.\end{cases}$$
Let us recall that by assumption functions $\psi_{i\mathrm{ext}} \   \forall i \in [1;N]$ are concave functions. Thus, function $\psi_{\mathrm{ext}}^j$ above defines a {\it perturbation function} which is a concave one; see \cite{boyd2004convex} for further details. Since $\psi_{\mathrm{ext}}^j\in\mathcal{K}_\infty\cup \{0\}$ are concave functions, there exists a concave function $\psi_{\mathrm{ext}}\in\mathcal{K}_\infty\cup \{0\}$ such that $\psi_{\mathrm{ext}}^j\leq\psi_{\mathrm{ext}}$ $\forall j\in \modes$. %By similar reasoning, using the assumption that $\psi_{i\mathrm{ext}} \   \forall i \in [1;N]$ are concave functions, we conclude that $\psi_{\mathrm{ext}}$ is a concave function. 
Hence, we conclude that $V$ is an SSF-M$_2$ function from $\hat{\Sigma}$ to $\Sigma$.
\end{proof}

\section{Example}
Consider an interconnection of $N \in \Natural$ stochastic hybrid subsystems $\Sigma_i$, $i \in [1;N]$, where each $\Sigma_i$ is given by $\Sigma_i = (\Real^{n_i},\Real^{n_i},\Real^{n_i},\mathcal{U}_i,\mathcal{W}_i,f_i,\sigma_i,\rho_i,\Real^{q_{1i}},\Real^{n_i},h_{1i},h_{2i})$, where for every $x_i \in \Real^{n_i}$, $u_i \in \Real^{n_i}$, $w_i \in \Real^{n_i}$:
\begin{align}
f_i(x_i,u_i,w_i) &= w_i + u_i, \nonumber \\
\sigma(x_i) &= \varpi x_i, \nonumber \\
\rho(x_i) &= \tau x_i, \nonumber \\
h_{1i}(x_i) &= C_{1i}x_i, \nonumber \\
h_{2i}(x_i) &= x_i,
\end{align}
where $\varpi \in \Real_{>0}$, $\tau \in \Real_{>0}$, and $C_{1i} \in \Real^{q_{1i}\times n_i}$. Each $\Sigma_i$ satisfies
\begin{IEEEeqnarray*}{c}
	\Sigma_i:\left\{
	\begin{IEEEeqnarraybox}[\relax][c]{rCl}
		\diff \xi_i(t) &=& (\omega_i(t) + \upsilon_i(t))\diff t + \varpi  \xi_i(t)\diff W_t + \tau \xi_i(t)\diff P_t\nonumber,\\
		\zeta_{1i}(t)&=& C_{1i}\xi_i(t), \\%
		\zeta_{2i}(t) &=& \xi_i(t).
	\end{IEEEeqnarraybox}\right.
\end{IEEEeqnarray*}
 Assume the rate of the Poisson process $P_t$ is $\lambda$. We consider a set of two interconnection topologies $M = \{M_1, M_2\}$ given by:
\begin{align}
M_1 &= -\frac{1}{n}\begin{bmatrix}
n - 1 & -1 & \dots & \dots & -1 \\
-1 & n-1 & -1 & \dots & -1 \\
-1 & -1 & n-1 & \dots & -1 \\
\vdots &  & \ddots & \ddots & \vdots \\
-1 & \dots & \dots & -1 & n-1
\end{bmatrix}, \nonumber \\
M_2 &= \begin{bmatrix}
-2 & 1 & 0 & 0 & \dots & 1 \\
1 & -2 & 1 & 0 & \dots & 0 \\
0 & 1 & -2 & 1 & \dots & 0 \\
\vdots & & & \ddots & & \\ 
& & & & \ddots & \\
1 & 0 & 0 & \dots & 1 & -2 \\
\end{bmatrix},
\end{align} 
where $n = \sum\limits_{i=1}^N n_i$, $M_1 \in \Real^{n \times n}$ represents a fully-connected interconnection topology, while $M_2 \in \Real^{n \times n}$ represents a cyclic interconnection topology. 
We consider a Markov chain $\Pi = (\modes,\mathsf{Q})$, with $\modes = \{1,2\}$ and $$\mathsf{Q} = \begin{bmatrix}
-0.5 & 0.5 \\
0.5 & -0.5
\end{bmatrix},$$
with the switching process $\pi$, which governs the switching between matrices $M_1$ and $M_2$. 
The interconnected switching stochastic hybrid system is denoted by $\mathcal{I}_{\pi}^M(\Sigma_1,\dots,\Sigma_N)$. We consider a scalar abstraction $\hat{\Sigma}_i = (\Real,\Real,\Real,\mathcal{\hat{U}}_i,\mathcal{\hat{W}}_i,\hat{f}_i,\hat{\sigma}_i,\hat{\rho}_i,\Real^{{q}_{1i}},\Real,\hat{h}_{1i},\hat{h}_{2i})$, where for every $\hat{x}_i \in \Real$, $\hat{u}_i \in \Real$, $\hat{w}_i \in \Real$:
\begin{align}
\hat{f}_i(\hat{x}_i,\hat{u}_i,\hat{w}_i) &= \hat{w}_i + \hat{u}_i, \nonumber \\
\hat{\sigma}(\hat{x}_i) &= {\varpi}\hat{x}_i , \nonumber \\
\hat{\rho}(\hat{x}_i) &= {\tau}\hat{x}_i , \nonumber \\
\hat{h}_{1i}(\hat{x}_i) &= C_{1i}\vec{1}_{n_i}\hat{x}_i, \nonumber \\
\hat{h}_{2i}(\hat{x}_i) &= \hat{x}_i.
\end{align}
%\begin{align}
%\hat{\Sigma}_i : 
%\begin{cases}
%\diff \hat{\xi}_i(t) = (\hat{\omega}_i(t)+ \hat{\upsilon}_i(t) + \varPhi_i(\hat{\xi}_i))\diff t \\
%\hat{\zeta}_{1i}(t) = C_{1i}\vec{1}_{n_i}\hat{\xi}_i(t) \nonumber \\
%\hat{\zeta}_{2i}(t) = \hat{\xi}_i(t),
%\end{cases}
%\end{align}
The function  $\mathcal{S}_i(x_i, \hat x_i) = (x_i-\vec{1}_{n_i}\hat{x}_i)^T(x_i-\vec{1}_{n_i}\hat{x}_i)$ is an SStF-M$_2$ from $\hat{\Sigma}_i$ to $\Sigma_i$, $\forall i \in [1;N]$, with the following parameters:
\begin{align}
& {\kappa}_i = 2\tilde{\chi} - 2\lambda\tau - \varpi^2 - \lambda\tau^2, W_i = I_{n_i}, X^{11}_i = 0_{n_i},\\ &  X^{22}_i= 0_{n_i}, X^{12}_i = X^{21}_i = I_{n_i},  H_i = \hat{W}_i = \vec{1}_{n_i},
\end{align}
for some $\tilde{\chi} > \lambda\tau + \frac{\varpi^2}{2} + \frac{\lambda\tau^2}{2}$, and with $\alpha_i(r) = \frac{1}{\lambda_{\max}(C_{1i}^TC_{1i})}r$, and $\psi_{i\mathsf{ext}}(r) = 0$, $\forall r \in \Real_{\geq 0}$. Input $u_i \in \Real^{n_i}$ is given via the so-called interface function \begin{align}
\label{eq:int1}
u_i = -\tilde{\chi}(x_i - \vec{1}_{n_i}\hat{x}_i) + \vec{1}_{n_i}\hat{u}_i. 
\end{align}
By selecting $\mu^j_1 = \ldots = \mu^j_N = 1$ for any $j \in \{1,2\}$, the function $V(x,\hat{x},j) = \sum^N_{i=1}\mu_i^j \mathcal{S}_i(x_i,\hat{x}_i)$ is an $\ssft$ function from $\hat{\Sigma}$ to $\Sigma$, where $\hat{\Sigma} = \mathcal{I}_{\pi}^{\hat{M}}(\hat{\Sigma}_1,\dots,\hat{\Sigma}_N)$ is the interconnection of abstract subsystems with a set of interconnection matrices $\hat{M} = \{\hat{M}_1,\hat{M}_2\}$, satisfying conditions \eqref{eq:interconnected1} and  \eqref{eq:interconnected2} for each $j \in \modes$. In this example, we choose $N$ = 3, $n_i = 50$, $C_{1i} = [1~0~\dots~0]$, $\forall i \in [1;N]$, $\varpi = 0.3$, $\tau = 0.03$, $\tilde{\chi} = 1,$ and $\lambda = 10$. The interconnection matrices for $\hat{\Sigma}$ are computed as:
\begin{align}
\hat{M}_1 &= 
\frac{1}{150}
\begin{bmatrix}
-100 & 50 & 50 \\
50 & -100 & 50 \\
50 & 50 & -100,
\end{bmatrix} \\
\hat{M}_2 &= \begin{bmatrix}
-2 & 1 & 1 \\
1 & -2 & 1 \\
1 & 1 & -2
\end{bmatrix}.
\end{align}
\subsection{Controller synthesis}

\begin{figure}
\hspace{-0.5cm}
	\centering
	\includegraphics[scale=0.5]{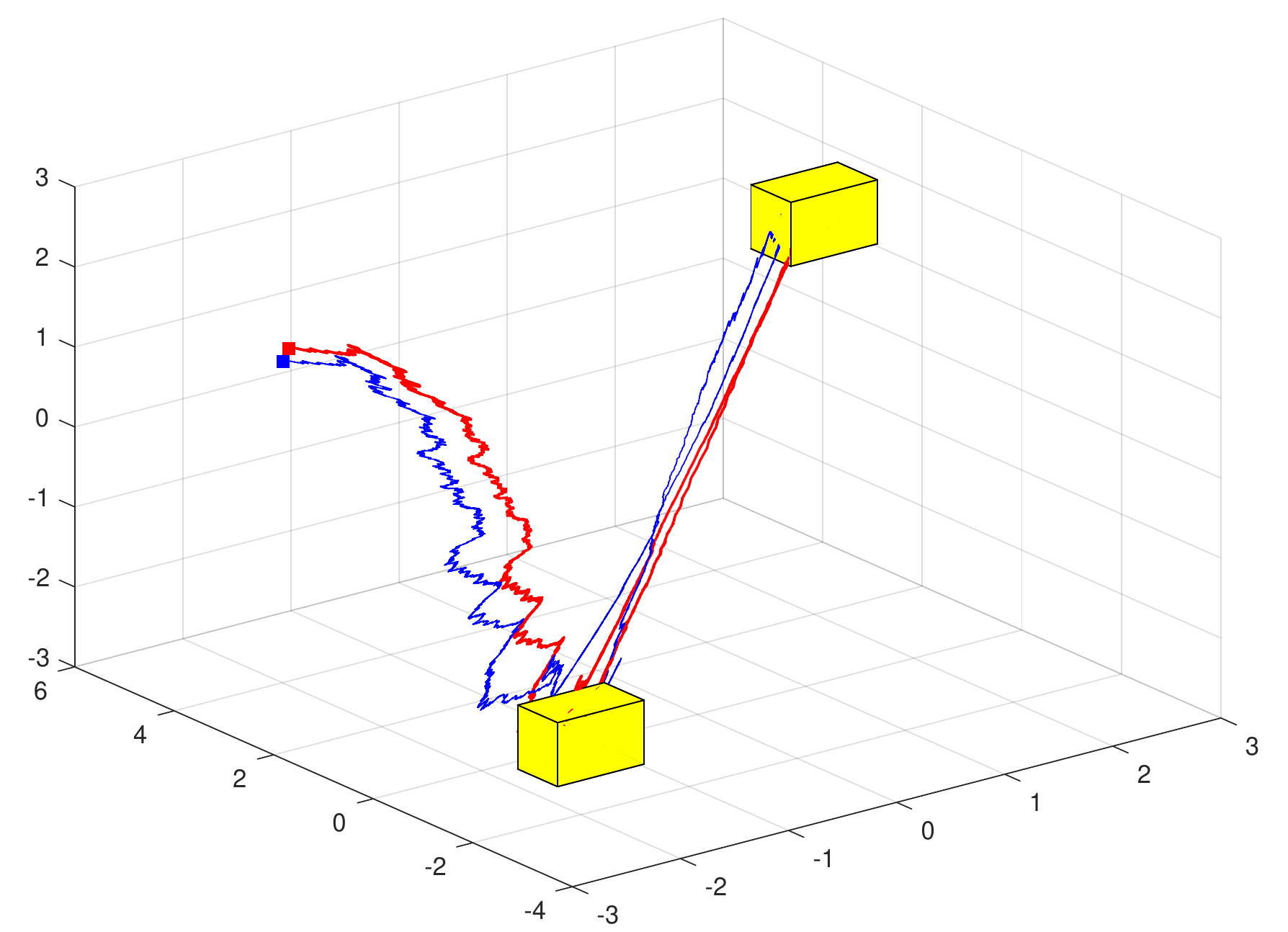}
	\caption{A realization of the output trajectories of the concrete (blue) and abstract (red) interconnected  stochastic hybrid systems with switched topologies ($\zeta(t)$ and $\hat{\zeta}(t)$ respectively). The yellow boxes indicate the two targets $T_1$ and $T_2$. The start points of the trajectories are indicated by the markers.}% \MZ{It seems the red trajectory is for $\tilde\Sigma$ (no noise), right?}
	\label{fig:fig1}
\end{figure}

\begin{figure}
	\hspace{-0.5cm}
	\centering
	\includegraphics[scale=0.6]{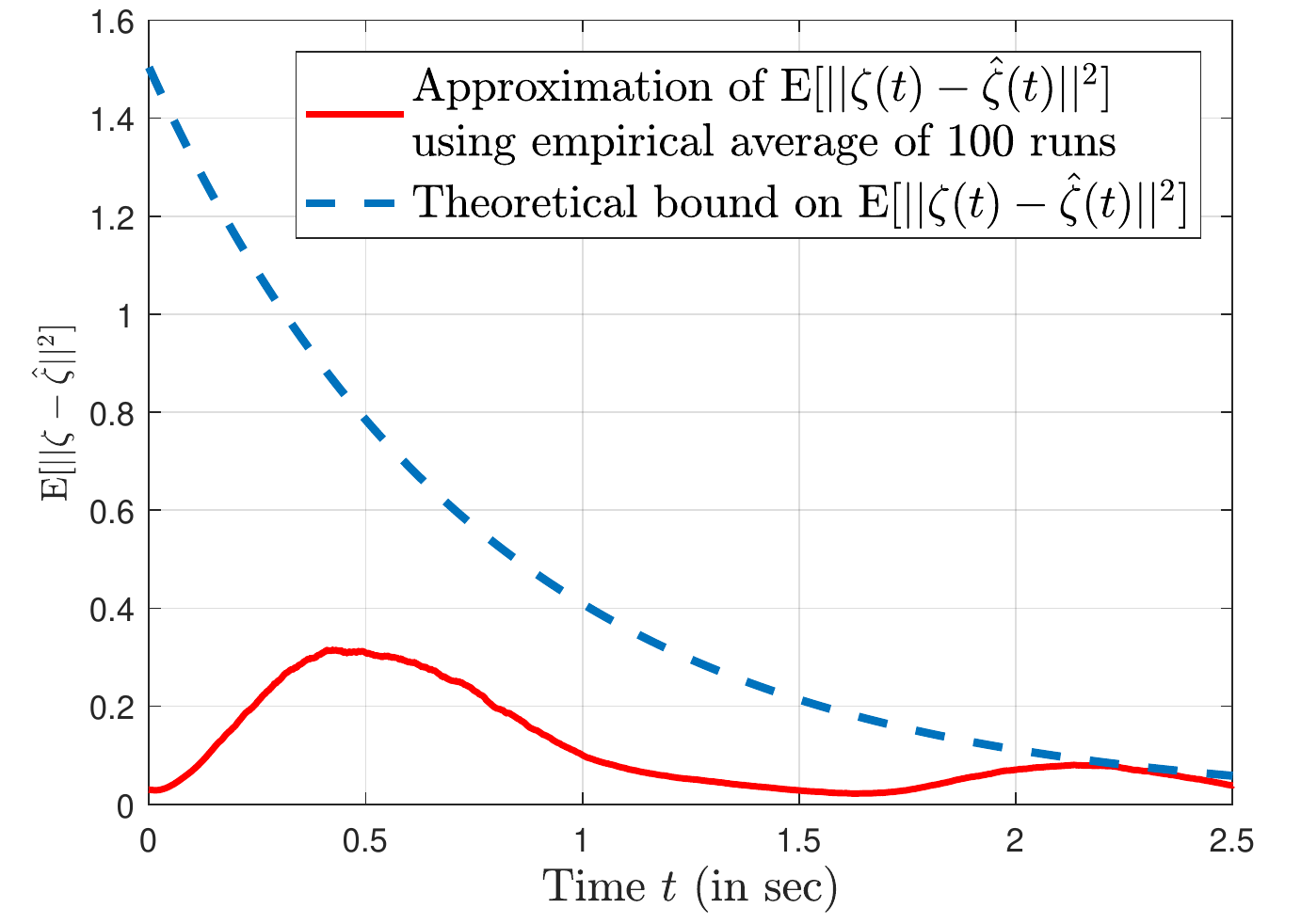}
	\caption{An approximation of $\Expectation{\norm{\zeta(t) - \hat{\zeta}(t)}^2}$ using empirical average of 100 runs and the theoretical upper bound obtained from \eqref{eq:bound_output} with $\beta(r,t) = r \expo^{-1.3t}$ and $\gamma_{\textrm{ext}}(r) = 0$.} %\MZ{Make the red trajectory thinner in which its noisiness is visible!}}
	\label{fig:fig2}
\end{figure}

In this sub-section we synthesize a controller for the abstract interconnected switching stochastic hybrid system $\hat{\Sigma} = \mathcal{I}_{\pi}^{\hat{M}}(\hat{\Sigma}_1,\dots,\hat{\Sigma}_N)$ to enforce a given specification, and then refine the designed controller for the interconnected switching stochastic hybrid system ${\Sigma} = \mathcal{I}_{\pi}^M({\Sigma}_1,\dots,{\Sigma}_N)$. First, we consider the randomly switched interconnected system $\tilde{\Sigma} = \mathcal{I}_{\pi}^{\hat{M}}(\tilde{\Sigma}_1,\dots,\tilde{\Sigma}_N)$, wherein each $\tilde{\Sigma}_i$, $i \in [1;N]$, results from $\hat{\Sigma}_i$ by setting the drift and reset terms to zero. 
It can be shown that function $V(\hat{x}, \tilde{x}) = \hat{x}^TQ_1\hat{x} + \tilde{x}^TQ_2\tilde{x}$, where $\hat{x} = [\hat{x}_1;\dots;\hat{x}_N]$ and $\tilde{x} = [\tilde{x}_1;\dots;\tilde{x}_N]$, is an $\ssft$ from $\tilde{\Sigma}$ to $\hat{\Sigma}$ with the associated interface function% \MZ{should not be $-\tilde{\chi}$ below?}
\begin{align}
\label{eq:int2}
\hat{u} &= -\tilde{\chi}(\hat{x}- \tilde{x}) + \tilde{u},
\end{align}
where $\hat{u} = [\hat{u}_1;\dots;\hat{u}_N]$ and $\tilde{u} = [\tilde{u}_1;\dots;\tilde{u}_N]$ and the matrices
$$Q_1 = \begin{bmatrix} 0.0708 & 0.0031 & 0.0031 \\
0.0031 & 0.0708 & 0.0031 \\
0.0031 & 0.0031 & 0.0708\end{bmatrix},$$
$$Q_2 = \begin{bmatrix} 2.9264 & -1.4392 & -1.4392 \\
-1.4392 & 2.9264 & -1.4392 \\
-1.4392 & -1.4392 & 2.9264\end{bmatrix},$$
which are obtained by solving an appropriate LMI. We synthesize a controller using toolbox \texttt{SCOTS} \cite{rungger2016scots} to synthesize a  controller to enforce the following linear temporal logic specification \cite{katoen08} over the outputs of $\tilde{\Sigma}$:
\begin{align}
\Psi = \square S \wedge \square \lozenge T_1 \wedge \square \lozenge T_2.
\end{align}
The property $\Psi$ can be interpreted as follows: the output trajectory $\tilde{\zeta}$ of the closed loop system evolves inside the set $S$, and visits $T_i$, $i \in [1;2]$, infinitely often, indicated with yellow boxes in Figure \ref{fig:fig1}. The sets $S$, $T_1$, and $T_2$ are given by: $S$ = $[-5,5]^3,$ $T_1 = [1.6,2.4]^3$, and $T_2 = [-2.4,-1.6]^3$. We use \eqref{eq:int1}
and \eqref{eq:int2} to generate the corresponding input enforcing this specification over original $\Sigma$. Figure \ref{fig:fig1} shows a realization of output trajectories $\Sigma$ and $\hat\Sigma$ started from initial conditions $ \zeta(0)=[-1.99;4.00;1.00]$ and $\hat{\zeta}(0)=[-1.89;4.10;1.10],$ respectively. %  \MZ{Can you plot several realizations of output errors together with the theoretical bound on one plot?}

\bibliographystyle{ieeetran}
\bibliography{bib}
\end{document}